\documentclass[12pt]{article}

\usepackage{amsmath,amssymb}
\usepackage{amsthm}
\usepackage{graphicx}
\usepackage{mathrsfs}
\usepackage{tikz}
\usepackage{booktabs}
\usepackage{floatflt}
\usepackage{enumerate}
\usepackage{psfrag}	
\usepackage{array}
\usepackage{multirow,hhline}
\usepackage{exscale}
\usepackage{color}
\usepackage{colortbl}
\usepackage{epsfig}

\newtheorem{theorem}{Theorem}

\newtheorem{lemma}{Lemma}
\newtheorem{definition}{Definition}

\begin{document}


\title{\textbf{Generalized Degrees of Freedom of the Interference Channel with a Signal Cognitive Relay}}

\author{\normalsize{Anas Chaaban and Aydin Sezgin}\\
\normalsize{Emmy-Noether Research Group on Wireless Networks}\footnote{This work is supported by the German Research Foundation, Deutsche Forschungsgemeinschaft (DFG), Germany, under grant SE 1697/3.}\\
\normalsize{Institute of Telecommunications and Applied Information Theory}\\
\normalsize{Ulm University, 89081, Ulm, Germany}\\
\normalsize{Email: {anas.chaaban@uni-ulm.de, aydin.sezgin@uni-ulm.de}}}

\maketitle

\begin{abstract}
We study the interference channel with a signal cognitive relay. A signal cognitive relay knows the transmit signals (but not the messages) of the sources non-causally, and tries to help them communicating with their respective destinations. We derive upper bounds and provide achievable schemes for this channel. These upper and lower bounds are shown to be tight from generalized degrees of freedom point of view. As a result, a characterization of the generalized degrees of freedom of the interference channel with a signal cognitive relay is given.
\end{abstract}

\section{Introduction}

Due to its increasing practical importance, the interference relay channel has witnessed an increasing research activity recently. The study of the relay channel was pioneered by \cite{CovGam79}, where the capacity of the point to point relay channel was analyzed. Recently, the interference relay channel (IC-R) has been the topic of many papers. 


Another aspect in information theory that has attracted research attention is the idea of cognition. The concept of cognitive radio aims at increasing the spectral efficiency of a system. The communicating nodes observe their medium and adapt their strategy, in such a way that improves the performance while avoiding interference with other nodes using the same medium. This opens new fields for study such as cognitive radio channels, and interference channels with a cognitive relay (IC-CR).

The cognitive radio channel has a been studied in \cite{DevMitTar} for example. The capacity of a class of cognitive radio channels was given recently in \cite{WuVishAra}, where the 2 user interference channel with a degraded message set is considered. Namely, the authors give the capacity of this channel under weak interference assumptions.

The IC-CR has also had its share of research interest. For example, the 2 user IC-CR is studied in \cite{SahErk}, which considers achievable rates in the IC-CR. A distinction between a signal cognitive relay (that knows the signals of both sources non causally, but not the messages) and message cognitive relay (that knows the messages of both sources non causally) is made in \cite{SahErk2}, where the 2 user IC-CR with one sided interference is analyzed.

Achievable rates and outer bounds for the 2 user IC-CR are given in \cite{SriVishJafSha}. The authors provide an achievable scheme that combines both Costa's dirty paper coding \cite{Costa} and Han-Kobayashi \cite{HanKob} schemes. Moreover, in the same paper, the authors derive the degrees of freedom (DOF) of this channel.

Besides the IC-CR, lots of work has been done on the IC with a full or a half duplex relay, see e.g. \cite{SahErkSim} and references therein. A couple of important resutls for the IC-R deserve special attention. In \cite{CadJaf}, the effect of relaying and other factors on the DOF of wireless networks is analyzed. The authors show that full duplex relays can not increase the DOF of a $K$ user X network even with feedback, cooperation, and full-duplex nodes. On the other hand, in an interesting setup in \cite{GomCadJaf}, the benefit of using a half-duplex relay in a $K$ user IC is discussed, where it was shown to be helpful for interference alignment. 

In addition to studying the capacity and the DOF of a channel, one might study the generalized DOF (GDOF). The idea of characterizing the GDOF was first used in \cite{Tse}, where the capacity of the IC is characterized to within one bit. Later, the GDOF of other channels were studied, e.g. \cite{AkuLev,HuaCadJaf}. In this paper, we study the GDOF of the symmetric IC with a signal cognitive relay (IC-SCR). We provide upper bounds on the achievable rates, and achievable schemes, and show that they coincide from a GDOF perspective.

The rest of the paper is organized as follows: In section \ref{Model}, we introduce the IC-SCR. In section \ref{Summary}, we summarize the results of this paper by stating the GDOF of the IC-SCR. We give upper bounds for the IC-SCR in section \ref{UpperBounds} and give GDOF achieving schemes in section \ref{GDOF}. Finally, we conclude in section \ref{Conc}.

\section{System Model}
\label{Model}
We consider a symmetric IC-SCR as shown in Figure \ref{model}. The transmitters and the relay have power $P$. The direct, cross, and relay channels are denoted $h_d$, $h_c$, and $h_r$ respectively. The relay is signal cognitive in the sense that it knows the transmit signals of both transmitters non-causally, but is not aware of the codebooks. The transmit signals of the IC-SCR can be written as follows:
\begin{align*}
x_1^n&=X_1^n\\
x_2^n&=X_2^n\\
x_r^n&=\frac{b_1}{a_1}X_1^n+\frac{b_2}{a_2}X_2^n
\end{align*}
where $\frac{1}{n}\sum_{i=1}^n\mathbb{E}[X_{j,i}^2]=a_j^2P$, $a_{j}^2\leq 1$ for $j\in\{1,2\}$. Since the power of the relay signal is $(b_1^2+b_2^2)P$, then we must have $\sum_{j=1}^2b_{j}^2\leq1$ to satisfy the power constraint at the relay. Thus, the received signals are given by
\begin{align}
\label{InOut}
Y_1=\left(h_d+\frac{b_1}{a_1}h_r\right)X_1+\left(h_c+\frac{b_2}{a_2}h_r\right)X_2+Z_1\nonumber\\
Y_2=\left(h_d+\frac{b_2}{a_2}h_r\right)X_2+\left(h_c+\frac{b_1}{a_1}h_r\right)X_1+Z_2.
\end{align}
We consider zero mean unit variance complex noise at both receivers, i.e. $Z_j\sim\mathcal{CN}(0,1)$.

\begin{figure}[h]
  \centering{
      \input{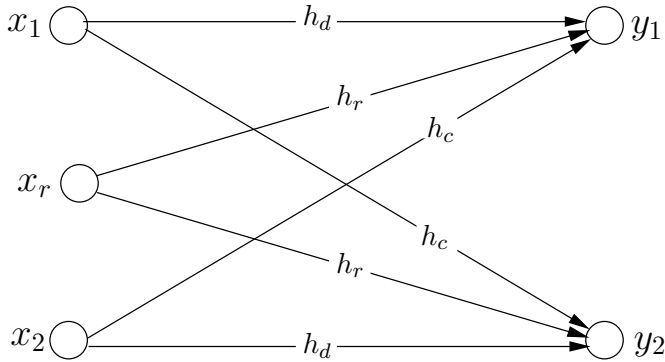}
    }
     \caption{The symmetric interference channel with a cognitive relay.}
     \label{model}
\end{figure}
There is no loss of generality in restricting the relay power to $P$, since any additional power at the relay can be modeled by larger $h_r$. Notice that in an IC-SCR, the received signals are similar to that of the IC, but with channel gains that are dependent on $a_j$ and $b_j$. This opens the possibility of artificially changing the channel parameters by properly choosing $a_j$ and $b_j$.

\section{Main Result}
\label{Summary}
In this section, we summarize the result of this paper. We give a characterization for the GDOF of the IC-SCR. For this purpose, we are going to need the following notations \cite{HuaCadJaf}:
\begin{align*}
\rho:=h_d^2P,\quad\alpha:=\lim_{\rho\to\infty}\frac{\log(h_c^2P)}{\log(\rho)},\quad\beta:=\lim_{\rho\to\infty}\frac{\log(h_r^2P)}{\log(\rho)},
\end{align*}
or equivalently
\begin{align*}
\lim_{\rho\to\infty}h_c^2P=\rho^\alpha,\quad \lim_{\rho\to\infty}h_r^2P=\rho^\beta.
\end{align*}
This definition varies slightly from the definitions in \cite{Tse}. As we will see, this definition will be more useful in our context.
\begin{definition}
We define the GDOF as 
\begin{align}
\label{def}
d(\alpha,\beta):=\lim_{\rho\to\infty}\frac{R_1+R_2}{2\log(\rho)}.
\end{align}
\end{definition}
Since we are considering the symmetric case, i.e. $d_1=d_2=d$, we divide by two in (\ref{def}) to obtain the GDOF per user.
Instead of $d(\alpha,\beta)$, we will use $d$ to denote the GDOF, which is clearly always a function of $\alpha$ and $\beta$. The GDOF of the IC-SCR is given in the following theorem.

\begin{theorem}
\label{MainTheorem}
The GDOF of the IC-SCR is given by:
\begin{align*}
d=\left\{\begin{array}{lr}
\min\left\{ \max\{1,\beta\},\max\left\{\beta,\frac{\alpha}{2}\right\},\max\left\{\alpha,\frac{\beta}{2}\right\} \right\},&\text{if } \alpha>1,\\
\\
\min\left\{\max\{1+\beta-\alpha,\alpha\},\max\left\{1-\frac{\alpha}{2},\beta,1+\beta-\alpha,1\right\}\right\}& \text{otherwise.}\end{array}\right.
\end{align*}
\end{theorem}
\begin{figure}[ht]
  \centering{
	\psfragscanon
	\psfrag{x}[][]{$\alpha$}
	\psfrag{y}[][]{$\beta$}
	\psfrag{o}[][]{$1$}
	\psfrag{b}[][]{$\beta$}
	\psfrag{a}[][]{$\alpha$}
	\psfrag{as}[][][1]{$\alpha$}
	\psfrag{1pbma}[c][][1][45]{$1+\beta-\alpha$}
	\psfrag{omab2}[][][1]{$1-\frac{\alpha}{2}$}
	\psfrag{a}[][]{$\alpha$}
	\psfrag{bb2}[][][1.3]{$\frac{\beta}{2}$}
	\psfrag{ab2}[][][1.3]{$\frac{\alpha}{2}$}
	\psfrag{os}[][][0.7]{\underline{1}}
	\psfrag{t}[][][0.7]{\underline{2}}
	\psfrag{th}[][][0.7]{\underline{3}}
	\psfrag{f}[][][0.7]{\underline{4}}
	\psfrag{fi}[][][0.7]{\underline{5}}
	\psfrag{s}[][][0.7]{\underline{6}}
	\psfrag{se}[][][0.7]{\underline{7}}
	\psfrag{e}[][][0.7]{\underline{8}}
	\psfrag{n}[][][0.7]{\underline{9}}
\psfragscanoff
      \includegraphics[width=0.8\textwidth]{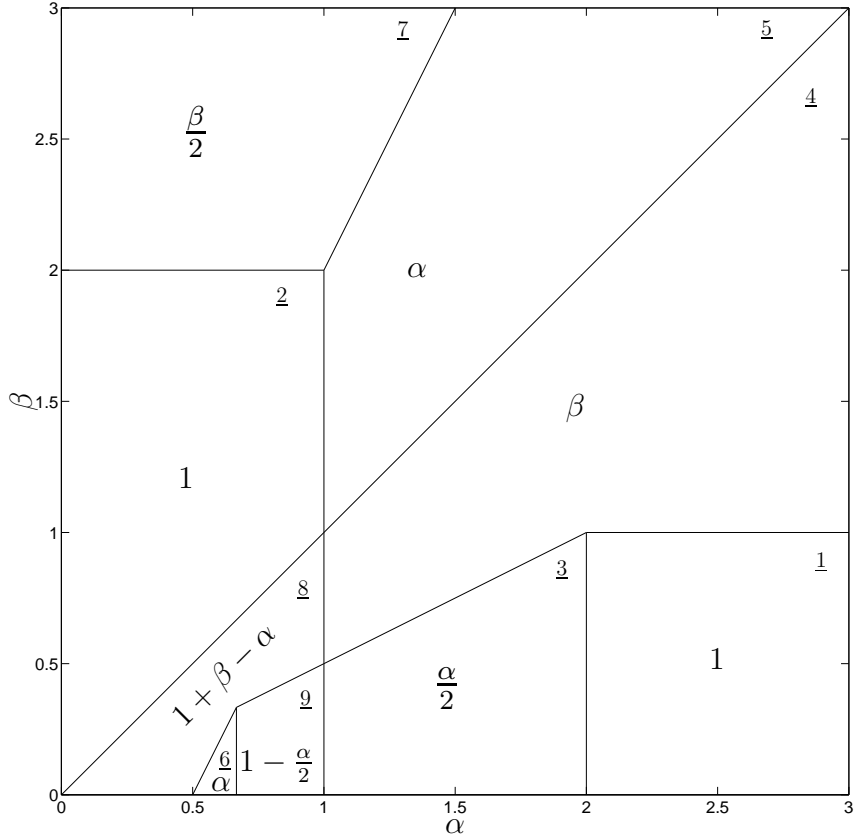}
    }
     \caption{The GDOF of the IC-SCR as a function of $\alpha$ and $\beta$.}
     \label{GDOFplane}
\end{figure}

For easier presentation of this characterization, we re-state it as follows
\begin{align*}
d=\left\{\begin{array}{lll}
1 & \text{Region \underline{1}}\\
1 &\text{Region \underline{2}} \\
\frac{\alpha}{2} & \text{Region \underline{3}} \\
\beta & \text{Region \underline{4}} \\ 
\alpha & \text{Region \underline{5}}\\
\alpha & \text{Region \underline{6}} \\
\frac{\beta}{2} & \text{Region \underline{7}} \\
1+\beta-\alpha & \text{Region \underline{8}} \\
1-\frac{\alpha}{2} & \text{Region \underline{9}}
\end{array}\right.
\end{align*}
Regions \underline{1}-\underline{9} are indicated by the numbered boxes in the Figure \ref{GDOFplane}, and are given by
\begin{enumerate}
\item[ \underline{1})] $\alpha\geq2  \text{ and }\beta\leq1$
\item[ \underline{2})] $\alpha\leq\beta\leq2\text{ and }\alpha\leq1$
\item[ \underline{3})] $\max\{1,2\beta\}\leq\alpha\leq2$ 
\item[ \underline{4})] $\alpha\geq\max\{1,\beta\}\text{ and }\beta\geq\min\left\{\frac{\alpha}{2},1\right\}$
\item[ \underline{5})] $\alpha\leq\beta\leq2\alpha\text{ and }\alpha\geq1$
\item[ \underline{6})] $\frac{\beta+1}{2}\leq\alpha\leq\frac{2}{3}$ 
\item[ \underline{7})] $\beta\geq\max\{2,2\alpha\}$ 
\item[ \underline{8})] $\min\left\{2\alpha-1,\frac{\alpha}{2}\right\}\leq\beta\leq\alpha \text{ and } \alpha\leq1$
\item[ \underline{9})] $\max\left\{\frac{2}{3},2\beta\right\}\leq\alpha\leq1$
\end{enumerate}
In regions \underline{1}, \underline{3}, \underline{6}, and \underline{9}, the IC-SCR behaves as an IC, where the relay does not have an impact on the GDOF. In this case, we can switch the relay off and still achieve the same GDOF using schemes from \cite{Tse}. However, in section \ref{GDOF}, we provide GDOF achieving schemes for these regions that incorporate the relay. 

On the other hand, in region \underline{7} where the relay channel is much stronger than both the direct and the cross channels, a BC behavior is noticed. Interestingly, for the remaining regions, an improvement over the IC and the BC is achieved by the cooperation between the sources and the signal cognitive relay.

We shortly list the achievable schemes, and discuss them in more details in section \ref{GDOF}.
\begin{itemize}
\item{Region \underline{1}:} Each receiver decodes the interference first, and then decodes his desired signal.
\item{Regions \underline{2}, \underline{4}, \underline{5}, and \underline{8}:} Sources and relay cooperate in a way that nulls the interference out.
\item{Regions \underline{3} and \underline{7}:} Each receiver decodes both signals.
\item{Regions \underline{6} and \underline{9}:} Han Kobayashi scheme.
\end{itemize}

It is worth to mention that this result for the symmetric GDOF simply generalizes to the non-symmetric case i.e. $d_1\neq d_2$. In general, we can simply multiply the symmetric GDOF $d$ by 2 to obtain a characterization of the sum of the GDOF for the non-symmetric case., where we can state $d_1+d_2=2d$.

\section{Converse to Theorem \ref{MainTheorem}}
\label{UpperBounds}
For the sake of characterizing the GDOF of the IC-SCR, we need to introduce some upper bounds. We will show later that these bounds are tight from GDOF perspective.
\begin{theorem}
\label{RateUB}
The rate region of the IC-SCR is upper bounded by
\begin{align*}
R_1&\leq C\left((h_d+h_r)^2P\right)\\
R_2&\leq C\left((h_d+h_r)^2P\right)\\
R_1+R_2&\leq\max\left\{2C\left(h_r^2\left(1-\frac{h_d}{h_c}\right)^2P\right),M_1,M_2\right\}&\text{if $h_c>\max\{h_d,h_r\}$}\\
R_1+R_2&\leq\max\left\{2C\Big((h_d-h_c)^2P\Big),M_1,M_2\right\}&\text{if $h_r>h_c>h_d$}\\
R_1+R_2&\leq\max\left\{\max_{\substack{a_1,a_2,b_1,b_2\\ h_c+c_jh_r\neq0}}C(q_4)+C(q_5),\max_{\substack{a_1,b_1,b_2\\
\phantom{x}}}C(q_4)+C(q_6)\right\}&\text{ if $h_c<h_d$}\\
R_1+R_2&\leq\max_{\substack{a_1,a_2,b_1,b_2}}C(q_7)+C(q_8) &\text{ if $h_c<h_d$}
\end{align*}
where
\begin{align*}
M_1&=\max_{\substack{b_1,b_2}}C(q_1)\\
M_2&=\max_{\substack{b_1,b_2}}\min\{C(q_2),C(q_3)\}\\
q_1&=b_1^2h_r^2\left(1-h_d/h_c\right)^2P+(h_c+b_2h_r)^2P\\
q_2&=(h_d+b_1h_r)^2P+(h_c+b_2h_r)^2P\\
q_3&=(h_d+b_2h_r)^2P+(h_c+b_1h_r)^2P\\
q_4&=(a_1h_d+b_1h_r)^2P\\
q_5&=\frac{(h_d+b_2h_r)^2P}{1+(a_1h_c+b_1h_r)^2P}\\
q_6&=\frac{b_2^2h_r^2(1-h_d/h_c)^2P}{1+(a_1h_c+b_1h_r)^2P}\\
q_7&=(a_2h_c+b_2h_r)^2P+\frac{(a_1h_d+b_1h_r)^2P}{1+(a_1h_c+b_1h_r)^2P}\\
q_8&=(a_1h_c+b_1h_r)^2P+\frac{(a_2h_d+b_2h_r)^2P}{1+(a_2h_c+b_2h_r)^2P},
\end{align*}
and $C(x)=\log(1+x)$.
\end{theorem}
\begin{proof}
See Appendix \ref{UB}.
\end{proof}

This rate upper bound can be written in terms of GDOF as follows
\begin{theorem}
The GDOF of the IC-SCR is upper bounded by
\begin{align}
\label{GDOF_SUB1}
d_1&\leq\max\{1,\beta\}&\forall \alpha,\beta\\
\label{GDOF_SUB2}
d_2&\leq\max\{1,\beta\}&\forall \alpha,\beta\\
\label{GDOF_MAC1}
d_1+d_2&\leq2\max\left\{\beta,\frac{\alpha}{2}\right\} &\text{if $\alpha>\max\{1,\beta\}$}\\
\label{GDOF_MAC2}
d_1+d_2&\leq2\max\left\{\frac{\beta}{2},\alpha\right\} &\text{if $\beta>\alpha>1$}\\
\label{GDOF_ZB}
d_1+d_2&\leq2\max\left\{1-\frac{\alpha}{2},\beta,1+\beta-\alpha,1\right\} &\text{if $\alpha<1$}\\
\label{GDOF_WIB}
d_1+d_2&\leq2\max\{1+\beta-\alpha,1\} &\text{if $\alpha<1$}
\end{align}
\end{theorem}
If we consider symmetric rates, i.e. $R_1=R_2$ or equivalently $d_1=d_2=d$, we can write this GDOF upper bound as
\begin{align*}
d\leq\left\{\begin{array}{lr}
\min\left\{ \max\{1,\beta\}, \max\left\{\beta,\frac{\alpha}{2}\right\},\max\left\{\alpha,\frac{\beta}{2}\right\} \right\},&\text{if } \alpha>1,\\
\\
\min\left\{\max\{1+\beta-\alpha,\alpha\},\max\left\{1-\frac{\alpha}{2},\beta,1+\beta-\alpha,1\right\}\right\}& \text{otherwise.}\end{array}\right.
\end{align*}
The regions where each term in the GDOF bound is active are given in Section \ref{Summary} and are enumerated as shown in Figure \ref{GDOFplane}. In what follows, we provide a GDOF achieving scheme for each of these regions.

\section{Achievability proof for Theorem \ref{MainTheorem}}
\label{GDOF}
In this section, we characterize the GDOF of the IC-SCR by providing GDOF achieving schemes. We characterize the regions in the following order: \underline{1}, \underline{2} and \underline{5}, \underline{4} and \underline{8}, \underline{3} and \underline{7}, \underline{6} and \underline{9}.

\subsection{Very Strong Interference (Region \underline{1})}
Recall that the very strong interference for the IC is defined for $\alpha>2$. In our case, for the IC-SCR, the very strong interference is defined for $\alpha>\max\{2,2\beta\}$.
We set $a_1=a_2=b_1=b_2=1/\sqrt2$. In this case, the received signals are
\begin{align*}
y_1^n=(h_d+h_r)x_1^n+(h_c+h_r)x_2^n+z_1^n\nonumber\\
y_2^n=(h_d+h_r)x_2^n+(h_c+h_r)x_1^n+z_2^n
\end{align*}
where $x_{j,i}\sim\mathcal{CN}(0,P/2)$, and $i\in\{1,\dots,n\}$ and $j\in\{1,2\}$. Receiver 1 decodes $x_2^n$ while treating his own signal as noise, and then decodes $x_1^n$. Similarly for receiver 2. As a result, the following rates are achieved
\begin{align*}
R_1&\leq\log\left(1+\frac{(h_d+h_r)^2P}{2}\right)\\
R_2&\leq\log\left(1+\frac{(h_d+h_r)^2P}{2}\right)
\end{align*}
as long as $I(x_2^n;y_1^n)\geq I(x_2^n;y_2^n|x_1^n)$ and $I(x_1^n;y_2^n)\geq I(x_1^n;y_1^n|x_2^n)$. In other words
\begin{align*}
\log\left(1+\frac{(h_c+h_r)^2P/2}{1+(h_d+h_r)^2P/2}\right)\geq\log(1+(h_d+h_r)^2P/2)
\end{align*}
i.e.
\begin{equation*}
(h_c+h_r)^2P\geq (h_d+h_r)^2P+(h_d+h_r)^4P^2/2.
\end{equation*}
And in terms GDOF, this gives
\begin{align*}
d\leq\max\{1,\beta\}, \text{ if } \alpha\geq\max\{2,2\beta\},
\end{align*}
achieving the single user bound given in (\ref{GDOF_SUB1}) if $\alpha\geq\max\{2,2\beta\}$. This scheme characterizes region \underline{1}.

\subsection{Zero Forcing (Regions \underline{2}, \underline{4}, \underline{5} and \underline{8})}
The relay and the transmitters can cooperate in such a way that cancels the interference completely. 
\subsubsection{$h_c\leq h_r$ (Regions \underline{2} and \underline{5})}
If $h_c\leq h_r$, the relay transmits $x_r^n=-\frac{h_c}{h_r}(x_1^n+x_2^n)$, where $x_{j,i}\sim\mathcal{CN}(0,P/2)$, $i\in\{1,\dots,n\}$, $j\in\{1,2\}$, and the sources transmit $x_1^n$ and $x_2^n$. Thus, the received signals are
\begin{align*}
y_1^n=(h_d-h_c)x_1^n+z_1^n,\nonumber\\
y_2^n=(h_d-h_c)x_2^n+z_2^n,
\end{align*}
and the achievable rates are given by
\begin{align*}
R_1&\leq\log\left(1+\frac{(h_d-h_c)^2P}{2}\right),\\
R_2&\leq\log\left(1+\frac{(h_d-h_c)^2P}{2}\right).
\end{align*}
Therefore, the achieved GDOF is
\begin{align*}
d\leq\max\{1,\alpha\}\quad \text{if } \alpha\leq\beta.
\end{align*}
This achieves the upper bound in (\ref{GDOF_SUB1}) if $1\geq\beta\geq\alpha$, and the upper bound in (\ref{GDOF_MAC2}) if $\max\{1,\alpha\}\leq\beta\leq \max\{2,2\alpha\}$. In other words, it achieves the upper bounds in regions \underline{2} and \underline{5}.

\subsubsection{$h_c\geq h_r$ (Regions \underline{4} and \underline{8})}
Otherwise, if $h_c\geq h_r$, the relay sends $x_r^n=x_1^n+x_2^n$, and the sources transmit $-\frac{h_r}{h_c}x_1^n$ and $-\frac{h_r}{h_c}x_2^n$. Thus, the received signals are given by
\begin{align*}
y_1^n=h_r\left(\frac{h_d}{h_c}-1\right)x_1^n+z_1^n,\nonumber\\
y_2^n=h_r\left(\frac{h_d}{h_c}-1\right)x_2^n+z_2^n,
\end{align*}
and the following rates are achieved
\begin{align*}
R_1&\leq\log\left(1+\left(\frac{h_d}{h_c}-1\right)^2\frac{h_r^2P}{2}\right),\\
R_2&\leq\log\left(1+\left(\frac{h_d}{h_c}-1\right)^2\frac{h_r^2P}{2}\right).
\end{align*}
Therefore, the achieved GDOF is
\begin{align*}
d\leq\max\{1+\beta-\alpha,\beta\}.
\end{align*}
This achieves  the upper bound in (\ref{GDOF_MAC1}) if ${1}/{2}\leq{\alpha}/{2}\leq\beta\leq\alpha$, the upper bound in (\ref{GDOF_ZB}) if $2\alpha-1\leq\beta\leq\alpha$, and the upper bound in (\ref{GDOF_WIB}) if $2\alpha-1\leq\beta\leq\alpha$. Hence, we characterize regions \underline{4} and \underline{8}.

\subsection{MAC (Regions \underline{3} and \underline{7})}
We set $a_1=a_2=b_1=b_2=1/\sqrt2$. In this case, the received signals are
\begin{align*}
y_1^n=(h_d+h_r)x_1^n+(h_c+h_r)x_2^n+z_1^n,\nonumber\\
y_2^n=(h_d+h_r)x_2^n+(h_c+h_r)x_1^n+z_2^n.
\end{align*}
The receivers jointly decode both signals $x_1^n$ and $x_2^n$ achieving
\begin{equation*}
R_1+R_2\leq\log(1+(h_d+h_r)^2P/2+(h_c+h_r)^2P/2).
\end{equation*}
Thus, the achievable GDOF is
\begin{align*}
d\leq\frac{1}{2}\max\{1,\alpha,\beta\}.
\end{align*}
This achieves the upper bound in (\ref{GDOF_MAC1}) if $\alpha\geq\max\{1,2\beta\}$, the upper bound in (\ref{GDOF_MAC2}) if $1\leq\alpha\leq\frac{\beta}{2}$, and the upper bound in (\ref{GDOF_ZB}) if $\beta\geq2$, and thus, characterizes regions \underline{3} and \underline{7}.

\subsection{Han-Kobayashi (Regions \underline{6} and \underline{9})}
The remaining regions (\underline{6} and \underline{9}) can be easily achieved by operating the IC-SCR as an IC using the Han Kobayashi scheme given in \cite{Tse}, i.e. with the relay is switched off. However, we provide a more general Han-Kobayashi scheme similar to \cite{Tse} that achieves the upper bounds for $\beta\leq\alpha\leq1$, which incorporates the relay. Let the transmit signals be given by:
\begin{align*}
x_1^n&=a(w_1^n+u_1^n),\\
x_2^n&=a(w_2^n+u_2^n).
\end{align*}
and let the relay transmit signal be:
\begin{align*}
x_r^n=\frac{b}{a}(x_1^n+x_2^n),
\end{align*}
where $w_{j,i}\sim\mathcal{CN}(0,P_w)$ carries a common message to be decoded by both receivers, and $u_{j,i}\sim\mathcal{CN}(0,P_u)$ carries a private message to be decoded by receiver $j$, such that $P_w+P_u=P/2$, $i\in\{1,\dots,n\}$ and $j\in\{1,2\}$. The receivers decode the common messages first while treating interference as noise. Thus, the rates of the common messages are bounded by
\begin{align*}
&\hspace{-0.6cm}R_{w1}+R_{w2}\leq\\
&\log\left(1+\frac{\big((ah_d+bh_r)^2+(ah_c+bh_r)^2\big)P_w}{1+\big((ah_d+bh_r)^2+(ah_c+bh_r)^2\big)P_u}\right).
\end{align*}
Moreover, let us fix $P_u$ such that $(ah_c+bh_r)^2P_u=1$, i.e. at the noise level. For readability, we will denote $(ah_d+bh_r)^2\frac{P}{2}$ by $S$ and $(ah_c+bh_r)^2\frac{P}{2}$ by $I$. Thus, we can write
\begin{align*}
R_{w1}+R_{w2}\leq\log\left(1+\frac{(S+I)(I-1)}{S+2I}\right).
\end{align*}
We also have the bound
\begin{align*}
R_{w1}+R_{w2}\leq2\log\left(1+\frac{I(I-1)}{S+2I}\right),
\end{align*}
as the sum of the individual rate constraints of decoding $w_1^n$ at receiver 2 and $w_2^n$ at receiver 1 while treating $u_1^n$ and $u_2^n$ as noise. The receivers decode the private message after canceling the effect of $w_1^n$ and $w_2^n$ from their received signal, while treating the undesired private messages as noise. Hence, the rate of the private message is given by
\begin{align*}
R_u\leq\log\left(1+\frac{S}{2I}\right)
\end{align*}
Thus, the symmetric rate can be written as 
\begin{align*}
&\hspace{-0.7cm}R\leq\min\left\{\log\left(1+I+\frac{S}{I}\right)-1,\right.\\
&\hspace{0.8cm}\left.\frac{1}{2}\log(1+S+I)+\frac{1}{2}\log\left(2+\frac{S}{I}\right)-1 \right\}.
\end{align*}
Notice that by setting $a=-b\frac{h_r}{h_c}$, we force $I$ to 0, and the system is interference free. In this case, we set $b=1$, and thus $S=h_r(\frac{h_d}{h_c}-1)P$. So, this scheme achieves the following GDOF
\begin{align*}
d\leq\min\left\{\max\left\{1-\frac{\alpha}{2},1+\beta-\alpha\right\},\max\left\{\alpha,1+\beta-\alpha\right\}\right\},
\end{align*}
where $d=1+\beta-\alpha$ is obtained if we choose $a=-b\frac{h_r}{h_c}$, otherwise, we obtain $d=1-\frac{\alpha}{2}$ or $d=\alpha$. This achieves the upper bound for $\beta\leq\alpha\leq1$ given in (\ref{GDOF_ZB}) and (\ref{GDOF_WIB}), and hence characterizes regions \underline{6} and \underline{9}. This scheme also provides an alternative for achieving the GDOF in region \underline{8}.

Figure \ref{GDOFalpha} shows the GDOF of the IC-SCR for different values of $\beta$. It can be seen in this figure that if $\beta>0$, the IC-CR has more GDOF compared to the IC. Moreover, $d$ is strictly increasing in $\beta$. 
\begin{figure}
  \centering{
	\psfragscanon
	\psfrag{x}[][]{$\alpha$}
	\psfrag{y}[][]{GDOF}
\psfragscanoff
      \includegraphics[width=0.75\textwidth]{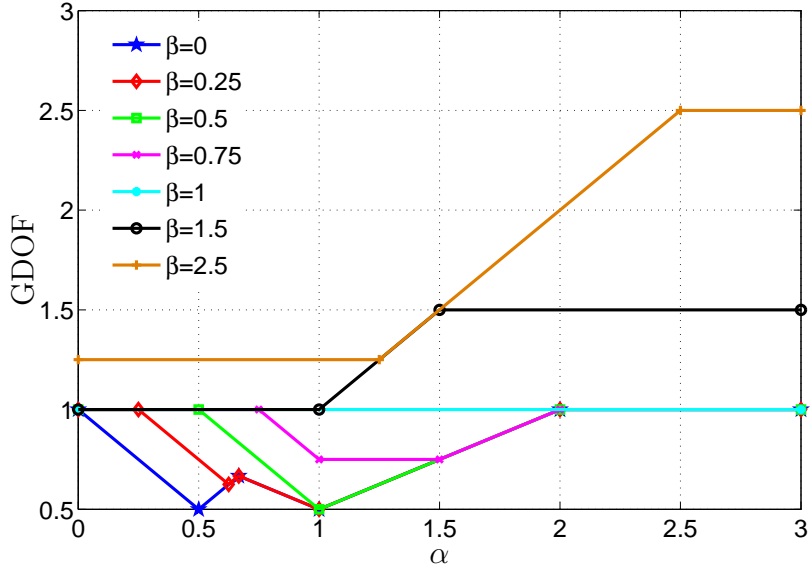}
    }
     \caption{The GDOF of the IC-SCR as a function of $\alpha$ for different values of $\beta$.}
     \label{GDOFalpha}
\end{figure}

\section{Conclusion}
\label{Conc}
A set of upper bounds for the the interference channel with a signal cognitive relay was given. These upper bounds give us a set of upper bounds on the generalized degrees of freedom of this channel.

We also provide achievable schemes for the channel under consideration, and show that these achievable schemes coincide with the upper bounds from a generalized degrees of freedom perspective. As a result, a characterization of the generalized degrees of freedom of the interference channel with a signal cognitive relay is provided.

Further extension of this work would be characterizing the generalized degrees of freedom of the interference channel with a message cognitive relay.

\begin{appendix}
\section{Proof of Theorem \ref{RateUB}}
We give here a proof of Theorem \ref{RateUB}. At the end of each subsection, the corresponding GDOF upper bound is given.
\label{UB}
\subsection{Single User Bound}
\label{SUB}
The single user bound follows from Fano's inequality
\begin{eqnarray*}
nR_1&\leq& I(X_1^n;Y_1^n|X_2^n)+n\epsilon_1,\\
nR_2&\leq& I(X_2^n;Y_2^n|X_1^n)+n\epsilon_2,
\end{eqnarray*}
where $\epsilon_1,\epsilon_2\to0$ as $n\to\infty$.
For $R_1$, we obtain
\begin{eqnarray*}
nR_1&\leq& h(Y_1^n|X_2^n)-h(Y_1^n|X_1^n,X_2^n)+n\epsilon_1\\
&\stackrel{(a)}{\leq}&\sum_{i=1}^nh(Y_{1,i}|X_{2,i})-h(Z_{1,i})+n\epsilon_1\\
&=&\sum_{i=1}^nh\left(\left(h_d+\frac{b_1}{a_1}h_r\right)X_{1,i}+Z_{1,i}\right)-h(Z_{1,i})+n\epsilon_1\\
&\stackrel{(b)}{\leq}&\sum_{i=1}^n\log\left(1+\left(h_d+\frac{b_1}{a_1}h_r\right)^2P_{1,i}\right)+n\epsilon_1\\
&\stackrel{(c)}{\leq}&n\log(1+(h_d+h_r)^2P)+n\epsilon_1,
\end{eqnarray*}
where (a) follows since conditioning does not increase entropy, (b) follows since the circularly symmetric complex Gaussian distribution maximizes differential entropy, and (c) follows from using Jensen's inequality applied to a function that is clearly concave, and from maximizing by setting $a_1^2=1$ and $b_1^2=1$. similarly for $R_2$, we obtain
\begin{eqnarray*}
nR_2&\leq&n\log(1+(h_d+h_r)^2P)+n\epsilon_2.
\end{eqnarray*}
Letting $n\to\infty$, we get
\begin{eqnarray*}
R_1&\leq&\log(1+(h_d+h_r)^2P),\\
R_2&\leq&\log(1+(h_d+h_r)^2P).
\end{eqnarray*}
In terms of GDOF, this is equivalent to
\begin{eqnarray*}
d&\leq&\lim_{\rho\to\infty}\frac{2\log(1+(\rho^{1/2}+\rho^{\beta/2})^2)}{2\log(\rho)}\\
&\approx&\lim_{\rho\to\infty}\frac{2\log(\max\{\rho,\rho^\beta\})}{2\log(\rho)}\\
&=&\max\{1,\beta\}.
\end{eqnarray*}

\subsection{The MAC bound for $h_c>h_d$ and $h_c>h_r$}
\label{MACB1}
As will be shown later, the MAC bound holds if $h_c>h_d$. Notice that the received signals of the IC-SCR are similar to that of the IC, with the difference that the channel coefficients can be controlled in our case by the choice of $a_j$ and $b_j$. Here, we distinguish between three cases:
\begin{itemize}
\item[i)] $h_c+\frac{b_1}{a_1}h_r=0$ and $h_c+\frac{b_2}{a_2}h_r= 0$,
\item[ii)] $h_c+\frac{b_1}{a_1}h_r=0$ and $h_c+\frac{b_2}{a_2}h_r\neq 0$,
\item[iii)] $h_c+\frac{b_1}{a_1}h_r\neq 0$ and $h_c+\frac{b_2}{a_2}h_r\neq 0$.
\end{itemize}
We obtain an upper bound for each case $R_1^i+R_2^i$, $R_1^{ii}+R_2^{ii}$, and $R_1^{iii}+R_2^{iii}$, and the upper bound that governs all cases is then
\begin{eqnarray*}
R_1+R_2\leq\max\{R_1^i+R_2^i,R_1^{ii}+R_2^{ii},R_1^{iii}+R_2^{iii}\}.
\end{eqnarray*}
The $\max$ operation follows since we can choose to be in case $i$, $ii$ or $iii$, and to maximize the upper bound, we choose to be in the case that results in a larger upper bound.

\begin{itemize}
\item[i)]
For case $i$, we have $h_c+\frac{b_j}{a_j}h_r=0$, for $j\in\{1,2\}$ i.e. 
\begin{eqnarray*}
\frac{b_j}{a_j}=-\frac{h_c}{h_r}.
\end{eqnarray*}
Thus, $a_j=-b_jh_r/h_c$. In this case, we have an interference free system, and the received signals become
\begin{eqnarray*}
Y_1^n=(h_d-h_c)X_1^n+Z_1^n,\\
Y_2^n=(h_d-h_c)X_2^n+Z_2^n.
\end{eqnarray*}
After maximizing over $b_j$ by setting $b_j^2=1$, we can upper bound the rates by
\begin{eqnarray*}
R_1^i&\leq&\log\left(1+\left(h_r-\frac{h_rh_d}{h_c}\right)^2P\right),\\
R_2^i&\leq&\log\left(1+\left(h_r-\frac{h_rh_d}{h_c}\right)^2P\right).
\end{eqnarray*}
Since $h_c>h_d$, i.e. $\alpha>1$, this corresponds to
\begin{eqnarray*}
d\leq\beta.
\end{eqnarray*}

\item[ii)]
In this case, only $h_c+\frac{b_1}{a_1}h_r=0$, so
\begin{eqnarray*}
\frac{b_1}{a_1}=-\frac{h_c}{h_r},
\end{eqnarray*}
i.e. $a_1=-b_1h_r/h_c$ and the received signals become
\begin{eqnarray*}
Y_1^n&=&(h_d-h_c)X_1^n+\left(h_c+\frac{b_2}{a_2}h_r\right)X_2^n+Z_1^n,\\
Y_2^n&=&\left(h_d+\frac{b_2}{a_2}h_r\right)X_2^n+Z_2^n.
\end{eqnarray*}
Consider a rate pair $(R_1^{ii},R_2^{ii})$ in the rate region of the IC-SCR. With this rate pair, both receivers are able to decode their desired signals. Notice that in this case, after decoding $X_1^n$, receiver 1 is able to decode $X_2^n$ as long as $h_c>h_d$, which is the case we are considering. It follows that the rate region in this case is included in the rate region of the MAC from $X_1^n$ and $X_2^n$ to $Y_1^n$. Thus, the sum rate is upper bounded by
\begin{eqnarray*}
R_1^{ii}+R_2^{ii}\leq\max_{\substack{b_1,b_2}}C(q_1),
\end{eqnarray*}
such that $b_1^2+b_2^2\leq1$, where
\begin{eqnarray*}
q_1=b_1^2\left(h_r-\frac{h_rh_d}{h_c}\right)^2P+(h_c+b_2h_r)^2P.
\end{eqnarray*}
Which translates to
\begin{eqnarray*}
d\leq\frac{1}{2}\max\{\beta,1+\beta-\alpha,\alpha\}=\frac{\alpha}{2},
\end{eqnarray*}
for $\alpha\geq\beta$ and $\alpha\geq1$.

\item[iii)]
If both $a_1h_c+b_1h_r\neq 0$ and $a_2h_c+b_2h_r\neq 0$, then both received signals suffer interference, and the received signals are as given in (\ref{InOut}). Consider a rate pair $(R_1^{iii},R_2^{iii})$ in the rate region of the IC-SCR. This means that receiver 1 is able to decode its desired signal, and then construct
\begin{eqnarray*}
\frac{h_d+\frac{b_2}{a_2}h_r}{h_c+\frac{b_2}{a_2}h_r}\left(Y_1^n-\left(h_d+\frac{b_1}{a_1}h_r\right)X_1^n\right)+\left(h_c+\frac{b_1}{a_1}h_r\right)X_1^n=Y_2^n-Z_2^n+\tilde{Z}_2^n,
\end{eqnarray*}
where $\tilde{Z}_2\sim\mathcal{CN}(0,(h_d+\frac{b_2}{a_2}h_r)^2(h_c+\frac{b_2}{a_2}h_r)^{-2})$. This is a less noisy version of $Y_2^n$ since $h_c>h_d$. Thus, receiver 1 is able to decode the $X_2$ after decoding $X_1$. Similarly for the second receiver. As a result, the rate region is included in that of the MAC from $X_1^n$ and $X_2^n$ to both $Y_1^n$ and $Y_2^n$. Therefore, the sum rate is bounded by:
\begin{eqnarray*}
R_1^{iii}+R_2^{iii}\leq&\max_{\substack{b_1,b_2}} \min\{C(q_2),C(q_3)\},
\end{eqnarray*}
such that $b_1^2+b_2^2\leq1$, where we use $a_1^2=a_2^2=1$ to maximize the bound with
\begin{eqnarray*}
q_2&=&(h_d+b_1h_r)^2P+(h_c+b_2h_r)^2P,\\
q_3&=&(h_d+b_2h_r)^2P+(h_c+b_1h_r)^2P.
\end{eqnarray*}
Or equivalently
\begin{eqnarray*}
d\leq\frac{1}{2}\max\{1,\alpha,\beta\}=\frac{\alpha}{2}
\end{eqnarray*}
\end{itemize}

As a result, we have the following sum rate upper bound
\begin{eqnarray*}
R_1+R_2&\leq&\max\{R_1^i+R_2^i,R_1^{ii}+R_2^{ii},R_1^{iii}+R_2^{iii}\}\\
&=&\max\left\{2\log\left(1+\left(h_r-\frac{h_rh_d}{h_c}\right)^2P\right),\max_{\substack{b_1,b_2}}C(q_1),\max_{\substack{b_1,b_2}}\min\{C(q_2),C(q_3)\}\right\},
\end{eqnarray*}
such that $b_1^2+b_2^2\leq1$.
Or in the GDOF language
\begin{eqnarray*}
d\leq\max\left\{\beta,\frac{\alpha}{2}\right\}.
\end{eqnarray*}

\subsection{The MAC bound for $h_d<h_c<h_r$}
\label{MACB2}
Similarly, it can be shown that the sum rate in this case is upper bounded by:
\begin{eqnarray*}
R_1+R_2&\leq&\max\left\{2\log(1+(h_d-h_c)^2P),\max_{\substack{b_1,b_2}}C(q_1),\max_{\substack{b_1,b_2}}\min\{C(q_2),C(q_3)\}\right\},
\end{eqnarray*}
such that $b_1^2+b_2^2\leq1$, i.e.
\begin{eqnarray*}
d\leq\max\left\{\alpha,\frac{\beta}{2}\right\}.
\end{eqnarray*}

\subsection{The Z channel bound}
\label{ZB}
We move our focus now to the case where $h_c<h_d$. For this bound, we are going to need the following lemma
\begin{lemma}[Worst Case Noise \cite{DigCov}]
\label{WCN}
Given a sequence of random variables $X^n$ with \\
$\sum_{j=1}^n\mathbb{E}(X^2_j)\leq~nP$, and a noise sequence $Z^n$ with $Z_j\sim\mathcal{CN}(0,\sigma^2)$, then $$h(X^n)-h(X^n+Z^n)\leq h(X_G^n)-h(X_G^n+Z^n).$$ where $X_G\sim\mathcal{CN}(0,P)$.
\end{lemma}
Now, let us assume that $h_c+\frac{b_j}{a_j}h_r\neq0$ for $j\in\{1,2\}$. In this case, we can upper bound the sum rate using the Z channel bound as follows
\begin{eqnarray*}
n(R_1+R_2)&\stackrel{(a)}{\leq}& I(X_1^n;Y_1^n,X_2^n)+I(X_2^n;Y_2^n)+n\epsilon\\
&\stackrel{(b)}{=}& I(X_1^n;Y_1^n|X_2^n)+I(X_2^n;Y_2^n)+n\epsilon\\
&=&h(Y_1^n|X_2^n)-h(Y_1^n|X_1^n,X_2^n)+h(Y_2^n)-h(Y_2^n|X_2^n)\\
&\stackrel{(c)}{\leq}&nh\left(\left(h_d+\frac{b_1}{a_1}h_r\right)X_{1G}+Z_1\right)-nh(Z_1)\\
&&+nh(Y_{2G})-nh\left(\left(h_c+\frac{b_1}{a_1}h_r\right)X_{1G}+Z_2\right)+n\epsilon\\
&\leq&n\log(1+(a_1h_d+b_1h_r)^2P)+n\log\left(1+\frac{(a_2h_d+b_2h_r)^2P}{1+(a_1h_c+b_1h_r)^2P}\right)+n\epsilon.
\end{eqnarray*}
where (a) follows from Fano's inequality, and (b) follows since $X_1$ and $X_2$ are independent. Using Lemma \ref{WCN}, we can show that 
$$h(Y_1^n|X_2^n)-h(Y_2^n|X_2^n)\leq nh\left(\left(h_d+\frac{b_1}{a_1}h_r\right)X_{1G}+Z_1\right)-h\left(\left(h_c+\frac{b_1}{a_1}h_r\right)X_{1G}+Z_2\right),$$ and we obtain (c), where $X_{jG}\sim\mathcal{CN}(0,a_j^2P)$. Clearly $h(Y_2^n)\leq nh(Y_{2G})$, where $Y_{2G}$ is the received signal at the second receiver when $X_{jG}$ is transmitted. As a result
\begin{eqnarray*}
R_1+R_2&\leq&\log(1+(a_1h_d+b_1h_r)^2P)+\log\left(1+\frac{(a_2h_d+b_2h_r)^2P}{1+(a_1h_c+b_1h_r)^2P}\right)\\
&\leq&\max_{\substack{a_1,a_2,b_1,b_2}}C(q_4)+C(q_5),
\end{eqnarray*}
with $a_j^2\leq1$, $\sum_{i=1}^2b_i^2\leq1$, and $h_c+\frac{b_j}{a_j}h_r\neq0$, where
\begin{eqnarray*}
q_4&=&(a_1h_d+b_1h_r)^2P\\
q_5&=&\frac{(h_d+b_2h_r)^2P}{1+(a_1h_c+b_1h_r)^2P}.
\end{eqnarray*}
In terms of GDOF, that is
\begin{eqnarray*}
d\leq \max\left\{1-\frac{\alpha}{2},\beta\right\}.
\end{eqnarray*}

On the other hand, assume $h_c+\frac{b_2}{a_2}h_r=0$, i.e. $a_2=-b_2h_r/h_c$. In this case the received signals are given by
\begin{eqnarray*}
Y_1^n&=&\left(h_d+\frac{b_1}{a_1}h_r\right)X_1^n+Z_1^n\\
Y_2^n&=&\left(h_c+\frac{b_1}{a_1}h_r\right)X_1^n+(h_d-h_c)X_2^n+Z_2^n.
\end{eqnarray*}
The rates can be bound as follows
\begin{eqnarray*}
n(R_1+R_2)&\stackrel{(a)}{\leq}&I(X_1^n;Y_1^n)+I(X_2^n;Y_2^n)+n\epsilon\\
&=&h(Y_1^n)-h(Y_1^n|X_1^n)+h(Y_2^n)-h(Y_2^n|X_2^n)+n\epsilon\\
&\stackrel{(b)}{\leq}&nh\left(\left(h_d+\frac{b_1}{a_1}h_r\right)X_{1G}+Z_1\right)-nh(Z_1)\\
&&+nh(Y_{2G})-nh\left(\left(h_c+\frac{b_1}{a_1}h_r\right)X_{1G}+Z_2\right)+n\epsilon\\
&=&n\log(1+(a_1h_d+b_1h_r)^2P)+n\log\left(1+\frac{(h_d-h_c)^2a_2^2P}{1+(a_1h_c+b_1h_r)^2P}\right)+n\epsilon\\
&=&n\log(1+(a_1h_d+b_1h_r)^2P)+n\log\left(1+\frac{b_2^2(h_r-h_dh_r/h_c)^2P}{1+(a_1h_c+b_1h_r)^2P}\right)+n\epsilon,
\end{eqnarray*}
where (a) follows from Fano's inequality, (b) from Lemma \ref{WCN}. Thus,
\begin{eqnarray*}
R_1+R_2\leq\max_{\substack{a_1,b_1,b_2}}C(q_4)+C(q_6),
\end{eqnarray*}
where $$q_6=\frac{b_2^2(h_r-h_dh_r/h_c)^2P}{1+(a_1h_c+b_1h_r)^2P}.$$
So
\begin{eqnarray*}
d\leq\max\{1+\beta-\alpha,1\}.
\end{eqnarray*}

Combining both bounds, we get the following upper bound
\begin{eqnarray*}
R_1+R_2\leq\max\left\{\max_{\substack{a_1,a_2,b_1,b_2\\ h_c+\frac{b_j}{a_j}h_r\neq0}}C(q_4)+C(q_5), \max_{\substack{a_1,b_1,b_2}}C(q_4)+C(q_6)\right\}
\end{eqnarray*}
such that $a_j^2\leq1$, and $\sum_{i=1}^2b_i^2\leq1$, i.e.
\begin{eqnarray*}
d\leq\max\left\{1-\frac{\alpha}{2},\beta,1+\beta-\alpha,1\right\}.
\end{eqnarray*}

\subsection{The Weak Interference Bound}
\label{WIB}
We still assume $h_c<h_d$. We use a genie similar to the one used in \cite{Tse}, let a genie give the signals
\begin{eqnarray*}
S_1^n&=&\left(h_c+\frac{b_1}{a_1}h_r\right)X_1^n+Z_2^n\\
S_2^n&=&\left(h_c+\frac{b_2}{a_2}h_r\right)X_2^n+Z_1^n
\end{eqnarray*}
to receivers 1 and 2 respectively. We can bound the rates as follows:
\begin{eqnarray*}
n(R_1+R_2)&\leq&I(X_1^n;Y_1^n,S_1^n)+I(X_2^n;Y_2^n,S_2^n)\\ 
&=&I(X_1^n;S_1^n)+I(X_1^n;Y_1^n|S_1^n)+I(X_2^n;S_2^n)+I(X_2^n;Y_2^n|S_2^n)\\
&=&h(S_1^n)-h(S_1^n|X_1^n)+h(Y_1^n|S_1^n)-h(Y_1^n|S_1^n,X_1^n)\\
&&+h(S_2^n)-h(S_2^n|X_2^n)+h(Y_2^n|S_2^n)-h(Y_2^n|S_2^n,X_2^n)\\
&=&h(S_1^n)-h(Z_2^n)+h(Y_1^n|S_1^n)-h(S_2^n)+h(S_2^n)-h(Z_1^n)+h(Y_2^n|S_2^n)-h(S_1^n)\\
&=&h(Y_1^n|S_1^n)+h(Y_2^n|S_2^n)-h(Z_1^n)-h(Z_2^n)\\
\end{eqnarray*}
Now we bound the expression $h(Y_1^n|S_1^n)-h(Z_1^n)$ as follows
\begin{eqnarray*}
h(Y_1^n|S_1^n)-h(Z_1^n)&\stackrel{(a)}{\leq}&\sum_{i=1}^n\log\left(1+\left(h_c+\frac{b_2}{a_2}h_r\right)^2P_{2,i}+\frac{(h_d+\frac{b_1}{a_1}h_r)^2P_{1,i}}{1+(h_c+\frac{b_1}{a_1}h_r)^2P_{1,i}}\right)\\
&\stackrel{(b)}{\leq}&n\log\left(1+(a_2h_c+b_2h_r)^2P+\frac{(a_1h_d+b_1h_r)^2P}{1+(a_1h_c+b_1h_r)^2P}\right),
\end{eqnarray*}
where (a) follows since the circularly symmetric complex Gaussian distribution maximizes the differential entropy, and (b) follows using Jensen's inequality. Similarly, we can bound $h(Y_2^n|S_2^n)-h(Z_2^n)$. It follows that
\begin{eqnarray*}
R_1+R_2\leq\max_{\substack{a_1,a_2,b_1,b_2}}C(q_7)+C(q_8),
\end{eqnarray*}
such that $a_1^2\leq1$, $a_2^2\leq1$, $b_1^2+b_2^2\leq1$, with
\begin{eqnarray*}
q_7&=&(a_2h_c+b_2h_r)^2P+\frac{(a_1h_d+b_1h_r)^2P}{1+(a_1h_c+b_1h_r)^2P}\\
q_8&=&(a_1h_c+b_1h_r)^2P+\frac{(a_2h_d+b_2h_r)^2P}{1+(a_2h_c+b_2h_r)^2P}.
\end{eqnarray*}
As a result
\begin{eqnarray*}
d\leq\max\{1+\beta-\alpha,1\}.
\end{eqnarray*}

\end{appendix}


\begin{thebibliography}{99}
\bibitem{CovGam79} T. Cover, A. E. Gamal, \emph{"Capacity Theorems for the Relay Channel"}, IEEE Trans. Inform. Theory, vol. 25, no. 5, pp. 572-584, 1979.
\bibitem{DevMitTar} N. Devroye, P. Mitran, V. Tarokh, \emph{"Achievable rates in cognitive radio channels"},IEEE Trans. Inform. Theory, vol. 52, no. 5, pp. 1813-1827.
\bibitem{WuVishAra} W. Wu, S. Vishwanath, A. Arapostathis, \emph{"Capacity of a Class of Cognitive Radio Channels: Interference Channels With Degraded Message Sets"}, IEEE Trans. Inform. Theory, Vol. 53, No. 11, pp. 4391-4399, November 2007.
\bibitem{SahErk} O. Sahin, E. Erkip, \emph{"On achievable Rates for Interference Relay Channel with Interference Cancelation"}, in  Proc. of 41$^{st}$ Annual Asilomar Conference on Signals, Systems and Computers, Pacific Grove, California, November 2007 (invited paper).
\bibitem{SahErkSim} O. Sahin, E. Erkip, O. Simeone, \emph{"Interference Channel with a Relay: Models, Relaying Strategies, Bounds"}, UCSD ITA Workshop 2009, San Diego, CA, February 2009.
\bibitem{SriVishJafSha} S. Sridharan, S. Vishwanath, S. A. Jafar, S. Shamai, \emph{"On the Capacity of Cognitive Relay assisted Gaussian Interference Channel"}, Proc. of IEEE International Symposium on Information Theory (ISIT) 2008.
\bibitem{Costa} M. Costa, \emph{"Writing on dirty paper"}, IEEE Trans. Inform. Theory, vol. IT-29, no. 3, pp. 439–441, May 1983.
\bibitem{HanKob} T. S. Han, K. Kobayashi, \emph{"A new achievable rate region for the interference channel"}, IEEE Trans. Inform. Theory, vol. IT-27, pp. 49-60, January 1981.
\bibitem{GomCadJaf} K. S. Gomadam, V. R. Cadambe, S. A. Jafar, \emph{"Approaching the capacity of wireless networks through distributed interference alignment"}, arXiv: 0803.3816, e-print.
\bibitem{SahErk2} O. Sahin, E. Erkip, \emph{"Cognitive Relaying with One-sided Interference"}, in Proc. of 42$^{nd}$ Annual Asilomar  Conference on Signals, Systems and Computers, Pacific Grove, California, October 2008 (invited).
\bibitem{CadJaf} V. R. Cadambe, S. A. Jafar, \emph{"Can Feedback, Cooperation, Relays and Full Duplex Operation Increase the Degrees of Freedom of Wireless Networks?"}, Proc. of IEEE International Symposium on Information Theory (ISIT) 2008.
\bibitem{Tse} R. H. Etkin, D. N. C. Tse, H. Wang, \emph{"Gaussian Interference Channel to Within One Bit"}, IEEE Trans. Inform. Theory, vol. 54, no. 12, pp. 5534-5562, December 2008.
\bibitem{AkuLev} E. Akuiyibo, O. Leveque, \emph{"Diversity-Multiplexing Tradeoff for the Slow Fading Interference Channel"}, Proc. of the 2008 International Zurich Seminar, March 2008.
\bibitem{HuaCadJaf} C. Huang, V. R. Cadambe, S. A. Jafar, \emph{"On the Capacity and Generalized Degrees of Freedom of the X Channel"}, e-print, arxiv: 0810.4741.
\bibitem{DigCov} S. Diggavi, T. Cover, \emph{"The Worst Additive Noise Under a Covariance Constraint"}, IEEE Trans. Inform. Theory, vol. 47, no. 7, November 2001.




\end{thebibliography}
\end{document}